%% file: main.tex
\documentclass[preprint,1p,11pt]{IR-Template/ISAS_IR}

\usepackage[cmex10]{amsmath}
\usepackage{amsfonts}
\usepackage[ruled]{algorithm2e}
\usepackage[english]{babel}
\usepackage{ISASPackages/ISASMatheMacros}
\usepackage{graphicx}
\usepackage{epstopdf}
\usepackage{subcaption}
\usepackage{flushend}
\usepackage{siunitx}
\usepackage{tabularx}
\usepackage{hyperref}

\usepackage{amsthm}
\newtheorem{definition}{Definition}
\newtheorem{lemma}{Lemma}

\begin{document}

\begin{frontmatter}

\title{Recursive Estimation of Orientation \\ Based on the Bingham Distribution}
\tnotetext[t1]{Draft submitted to 16th International Conference on Information FUSION on March 15, 2013}

\author[isas]{Gerhard~Kurz}
\ead{gerhard.kurz@kit.edu}

\author[isas]{Igor~Gilitschenski}
\ead{gilitschenski@kit.edu}

\author[ucl]{Simon~Julier}
\ead{s.julier@cs.ucl.ac.uk}

\author[isas]{Uwe~D.~Hanebeck}
\ead{uwe.hanebeck@ieee.org}

\address[isas]{Intelligent Sensor-Actuator-Systems Laboratory (ISAS)\\
Institute for Anthropomatics\\
Karlsruhe Institute of Technology (KIT), Germany\vspace{3mm}}

\address[ucl]{Virtual Environments and Computer Graphics Group\\
Department of Computer Science\\
University College London (UCL), United Kingdom}

\begin{abstract}
Directional estimation is a common problem in many tracking applications. Traditional filters such as the Kalman filter perform poorly because they fail to take the periodic nature of the problem into account. We present a recursive filter for directional data based on the Bingham distribution in two dimensions. The proposed filter can be applied to circular filtering problems with 180 degree symmetry, i.e., rotations by 180 degrees cannot be distinguished. It is easily implemented using standard numerical techniques and suitable for real-time applications. The presented approach is extensible to quaternions, which allow tracking arbitrary three-dimensional orientations. We evaluate our filter in a challenging scenario and compare it to a traditional Kalman filtering approach.
\end{abstract}

\end{frontmatter}

\section{Introduction} \label{sec:introduction}
\noindent
\input{introduction}

\section{Related Work} \label{sec:relatedwork}
\noindent
\input{related_work}

\section{Key Idea of the Bingham Filter} \label{sec:keyidea}
\noindent
\input{key_idea}

\section{Bingham Distribution} \label{sec:binghamdistribution}
\noindent
\input{bingham_distribution}

\section{Filter Implementation} \label{sec:filter}
\noindent
\input{filter}

\section{Evaluation} \label{sec:evaluation}
\noindent
\input{evaluation}

\section{Conclusion} \label{sec:conclusion}
\noindent
\input{conclusion}

\section*{Acknowledgment} \noindent
This work was partially supported by grants from the German Research Foundation
(DFG) within the Research Training Groups RTG 1194 ``Self-organizing
Sensor-Actuator-Networks'' and RTG 1126 ``Soft-tissue Surgery: New
Computer-based Methods for the Future Workplace''.


\begin{appendix}

\input{appendix.tex}

\end{appendix}
    

\bibliographystyle{ieeetemplates/IEEEtran}
\bibliography{literature-ig,gk-bingham}


\end{document}

%% file: introduction.tex
Many estimation problems involve the task of estimating angular values. These problems include, but are not limited to, estimating the pose or orientation of objects. For example, tracking cars, ships, or airplanes may involve estimation of their current orientation or heading. Furthermore, many applications in the area of robotics or augmented reality depend on reliable estimation of the pose of certain objects. When estimating the orientation of two-way roads or relative angles of two unlabeled targets, the estimation task reduces to estimating an axis. This can be thought of as estimation of a directionless orientation or estimation with $180^\circ$ symmetry. All these estimation problems share the need for processing angular or directional data, which differs in many ways from the classical Euclidean setting. First, periodicity needs to be taken into account. This is especially important for measurement updates around $0$, respectively $2\pi$. Second, directional quantities do not lie in a vector space. Thus, there is no equivalent to a classical linear model, as there are no linear mappings. 

In many current applications, even simple estimation problems involving angular data are often considered as traditional linear or nonlinear estimation problems and handled with classical techniques such as the Kalman Filter \cite{kalman1960}, the extended Kalman Filter (EKF), or the unscented Kalman Filter (UKF) \cite{julier2004}. In a circular setting, most traditional approaches to filtering suffer from assuming a Gaussian probability density at a certain point. They fail to take into account the periodic nature of the problem and assume a linear vector space instead of a curved manifold. This shortcoming can cause poor results, in particular when the angular uncertainty is large. In certain cases, the filter may even diverge. 

Classical strategies to avoid these problems in an angular setting involve an ``intelligent'' repositioning of measurements or even discarding certain undesired measurements. Sometimes, nonlinear equality constraints have to be fulfilled, for example unit length of a vector, which makes it necessary to inflate the covariance \cite{julier2007}. There are also approaches, that use operators on a manifold to provide a local approximation of a vector space \cite{Hertzberg2013}. While these approaches yield feasible results, they still suffer from ignoring the true geometry of circular data within their probabilistic models, which are usually based on assuming a normally distributed noise. This assumption is often motivated by the Central Limit Theorem, i.e., the limit distribution of a normalized sum of i.i.d. random variables with finite variance is normally distributed \cite{shiryaev1995}. For angular data, this is not the case. Choosing a circular distribution for describing uncertainty offers possibly better results. 

In this paper, we consider the use of the Bingham distribution \cite{bingham1974} for recursive estimation of orientation. The Bingham distribution is defined on the hypersphere of arbitrary dimension, so it can be applied to problems of different dimensionality. Here, we focus on the two-dimensional case and apply our results to axis estimation. To the best of our knowledge, this is the first published attempt to create a recursive filter based on the Bingham distribution. 

The presented methods can also be applied to the four-dimensional case, which would allow the representation of unit quaternions. Unit quaternions could then be used to estimate the full 3D orientation of an object. It is well known that Quaternions avoid the singularities present in other representations such as Euler angles. Their only downsides are the fact that they must remain normalized and the property that there are two quaternions for every rotation ($q$ and $-q$). Both of these issues can elegantly be overcome by use of the Bingham distribution, since it is by definition restricted to the hypersphere and is $180^\circ$ symmetric.


This paper is structured as follows. First, we present an overview of previous work in the area of directional statistics and angular estimation (Sec. \ref{sec:relatedwork}). Then, we introduce our key idea in Sec. \ref{sec:keyidea}. In Sec. \ref{sec:binghamdistribution}, we give a detailed introduction to the Bingham distribution and we derive the necessary operations, which we will need to create a recursive Bingham filter. Based on these prerequisites, we introduce our filter in Sec. \ref{sec:filter}. We have carried out an evaluation in simulations, which is presented in Sec. \ref{sec:evaluation}. Finally, we conclude this work in Sec. \ref{sec:conclusion}.


%% file: related_work.tex
Directional statistics is a subdiscipline of statistics, which focuses on dealing with directional data. Classical results in directional statistics are summed up in the books by Mardia and Jupp \cite{mardia1999} and by Jammalamadaka and Sengupta \cite{jammalamadaka2001}. Directional statistics differs from traditional statistics by the fact that random variables located on manifolds (for example the circle or the sphere) are considered rather than random variables located in vector spaces (typically $\mathbb{R}^d$).

There is a broad range of research for investigating the 2D orientation, e.\,g., the work by Krindis et al. \cite{krinidis2006}. A recursive filter based on the von Mises distribution for estimating the orientation on the $SO(2)$ was presented in \cite{azmani2009}. Later, a nonlinear filter based on von Mises and wrapped normal distributions was presented in \cite{kurz2013}.

In 1974, Bingham proposed his distribution in \cite{bingham1974}. Further work on the Bingham distribution has been done by Kent \cite{kent1987} as well as Jupp and Mardia \cite{jupp1979}. So far, there have only been a few applications of the Bingham distribution, for example in geology \cite{kunze2004}. In 2011, Glover used the Bingham distribution for a Monte Carlo based pose estimation \cite{glover2011}. Glover also released a library called libbingham \cite{Glover13} that includes implementations of some of the methods discussed in Sec. \ref{sec:binghamdistribution}. It should be noted that our implementation is not based on libbingham.



%% file: key_idea.tex
The goal of this paper is the derivation of a recursive filter based on the Bingham distribution. Rather than relying on the traditional Gaussian distribution, we chose to represent all occurring probability densities as Bingham. The Bingham distribution is defined on the hypersphere and is antipodally symmetric, which makes it interesting for applications in angular estimation with inherent $180^\circ$ symmetry and for problems where $180^\circ$ symmetry occurs as a result of parameterization, e.g., in the case of quaternions. Although we restrict ourselves to the two-dimensional case in this paper, we would like to emphasize that most of presented methods are easily generalized to higher dimensions.

In order to derive a recursive filter, we need to be able to perform two operations. First, we need to calculate the predicted state at the next time step from the current state and the system noise affecting the state. In a traditional estimation problem in $\mathbb{R}^d$ with additive noise, this involves a convolution with the noise density. We provide a suitable analogue on the hypersphere, which we call \emph{composition}. Since Bingham distributions are not closed under compositions, we present an approximate solution to this problem, which is based on matching covariance matrices. 

Second, we need to perform a Bayes update. As usual, this requires the \emph{multiplication} of the prior density with the likelihood density. We prove that Bingham distributions are closed under multiplication and show how to obtain the posterior density.

%% file: bingham_distribution.tex
The Bingham distribution appears naturally when a $d$-dimensional normal random vector ${\bf \vec{x}}$ with $\E({\bf \vec{x}})=\vec{0}$ is conditioned on $||{\bf \vec{x}}||=1$ \cite{kume2005}. In the following, we will introduce the Bingham distribution and derive the formulas for multiplication of two Bingham probability density functions. Furthermore, we will present a method for computing the composition of two Bingham-distributed random variables, which is analogous to the addition of real random variables.

\subsection{Probability Density Function}
\begin{definition}
Let $S_{d-1} = \{ \vec{x} \in \mathbb{R}^d : ||\vec{x}||=1 \} \subset \mathbb{R}^d$ be the unit hypersphere in $\mathbb{R}^d$. The probability density function (pdf) 
\[ f: S_{d-1} \to \mathbb{R} \]
of a Bingham  distribution \cite{bingham1974} is given by
\[ f(\vec{x}) = \frac{1}{F} \cdot \exp ( \vec{x}^T \mat{M}\, \mat{Z}\, \mat{M}^T \vec{x} )\ ,  \]
where $\mat{M} \in \mathbb{R}^{d \times d}$ is an orthogonal matrix ($\mat{M}\,\mat{M}^T = \mat{M}^T\,\mat{M} = \mat{I}_{d \times d}$) describing the orientation, $\mat{Z} = \text{diag} (z_1, \dots z_{d-1}, 0) \in \mathbb{R}^{d \times d}$ with $z_1 \leq \dots \leq z_{d-1} \leq 0$ is the concentration matrix, and $F$ is a normalization constant. 
\end{definition}
As Bingham showed, adding a multiple of the identity matrix $ \mat{I}_{d \times d}$ to $\mat{Z}$ does not change the distribution. Thus, we conveniently force the last entry of \mat{Z} to be zero. Because it is possible to swap columns of $\mat{M}$ and the according diagonal entries in $\mat{Z}$ without changing the distribution, we can enforce $z_1 \leq \dots \leq z_{d-1}$. This representation allows us to obtain the mode of the distribution very easily by taking the last column of $\mat{M}$.

The pdf is antipodally symmetric, i.\,e., $f(\vec{x}) = f(- \vec{x})$ holds for all $x \in S_{d-1}$. Consequently, the Bingham distribution is invariant to rotations by $180^\circ$. Examples of the pdf for two dimensions ($d=2$) are shown in Fig. \ref{fig:pdf3d} and Fig. \ref{fig:pdf2d}. The Bingham distribution is very similar to a Gaussian if and only if the uncertainty is small. This can be seen in Fig. \ref{fig:kld}, which shows the Kullback-Leibler divergence between a Bingham pdf and a corresponding Gaussian pdf.  

\begin{figure}[t]
\centering
\includegraphics[width=0.6\textwidth]{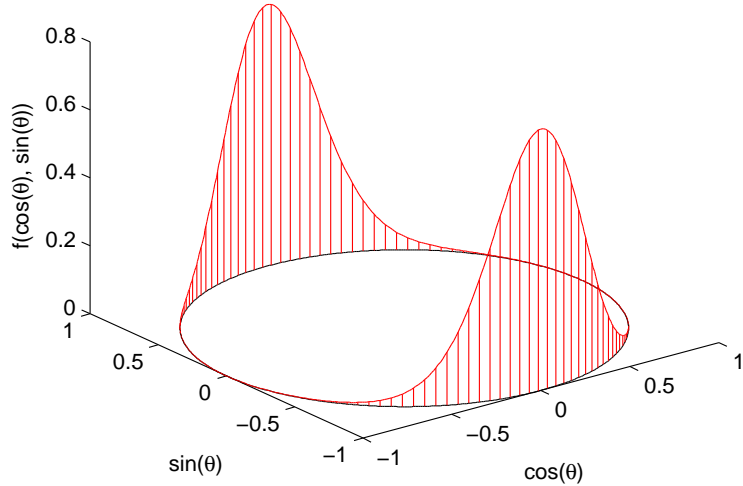}
\caption{Bingham pdf with  $\mat{M}=\mat{I}_{2 \times 2}$ and $\mat{Z} = \diag(-8, 0)$ as a 3D plot. This corresponds to a standard deviation of $16^\circ$.}
\label{fig:pdf3d}
\end{figure}

\begin{figure}[t]
\centering
\includegraphics[width=0.6\textwidth]{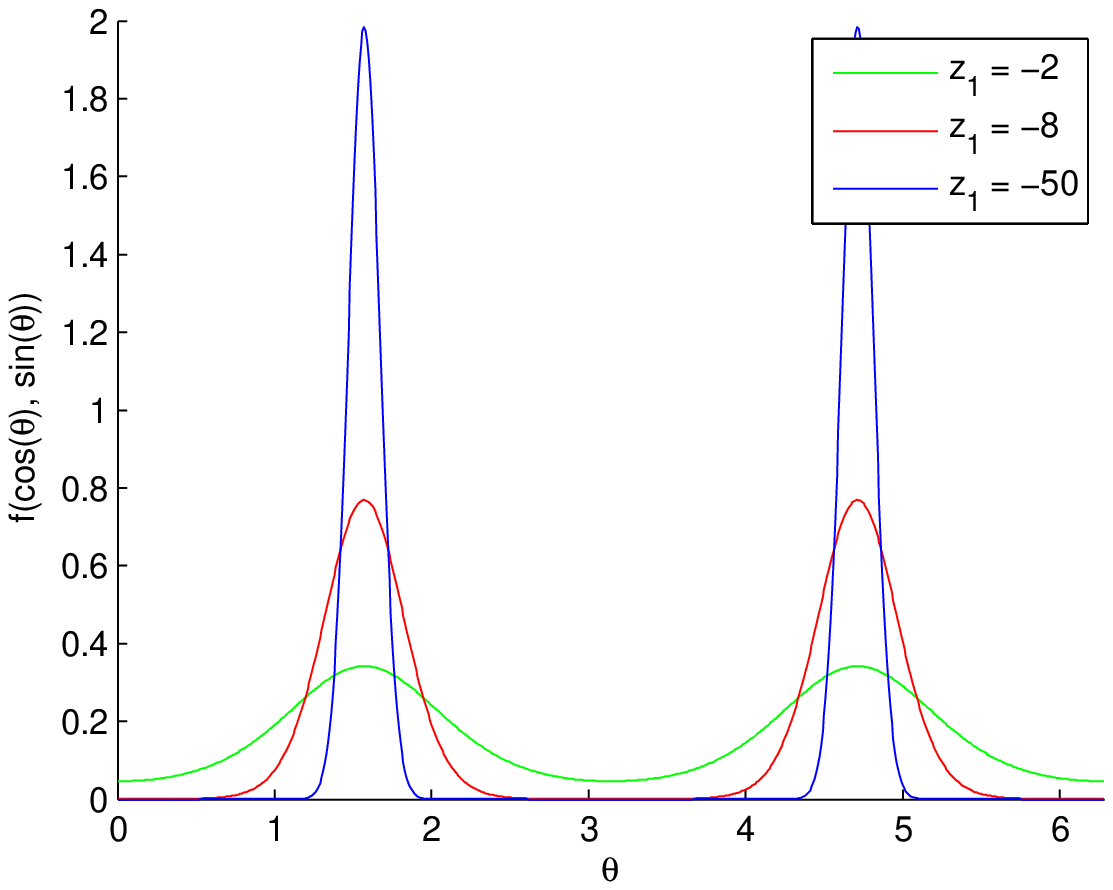}
\caption{Bingham pdf with $\mat{M}=\mat{I}_{2 \times 2}$ for different values of $\mat{Z} = \diag (z_1, 0)$ and $x=(\cos(\theta), \sin(\theta))^T$. These values for $z_1$ correspond to standard deviations of approximately $36^\circ$, $16^\circ$, and $6^\circ$ respectively.}
\label{fig:pdf2d}
\end{figure}

\begin{figure}[t]
\centering
\includegraphics[width=0.6\textwidth]{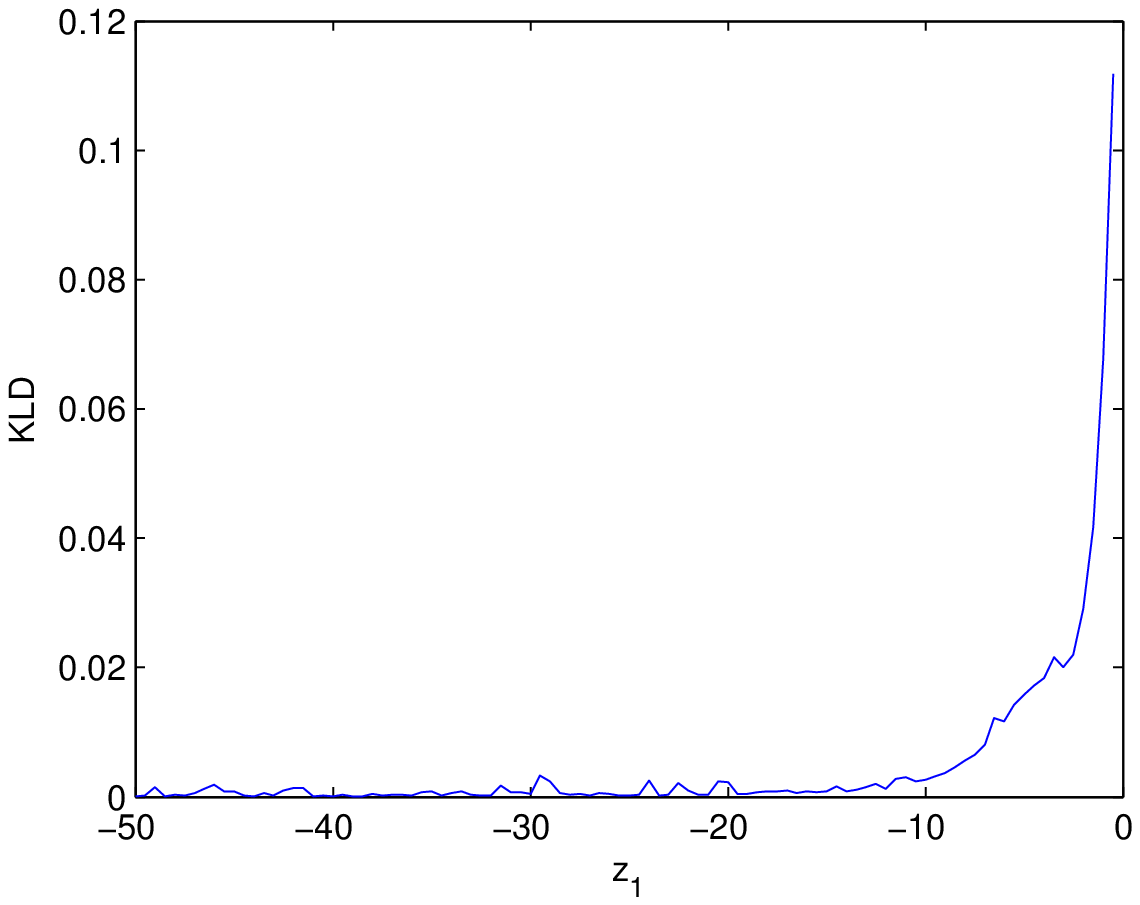}
\caption{Kullback-Leibler divergence on the interval $[0, \pi]$ between a Bingham pdf with $\mat{M}=\mat{I}_{2 \times 2}$, $\mat{Z} = \diag (z_1, 0)$ and a Gaussian pdf with equal mode and standard deviation. For small uncertainties ($z_1 < - 15$, which corresponds to a standard deviation of about $11^\circ$), the Gaussian and Bingham distributions are almost indistinguishable. However, for large uncertainties, the Gaussian approximation becomes quite poor.}
\label{fig:kld}
\end{figure}


\subsection{Normalization Constant}
The normalization constant can be calculated with the help of the hypergeometric function of a matrix argument \cite{herz1955,koev2006,muirhead1982}. It is given by
\[ F := |S_{d-1}| \cdot {}_1F_1 \left( \frac{1}{2}, \frac{d}{2}, \mat{Z} \right)\ , \]
where $|S_{d-1}|$ is the surface area of the $d$-sphere and ${}_1F_1(\cdot,\cdot,\cdot)$ is the hypergeometric function of matrix argument. In the two-dimensional case ($d=2$), this reduces to
\[ F = 2 \pi \cdot {}_1F_1 \left( \frac{1}{2}, 1, \begin{pmatrix} z_1 & 0 \\ 0 & 0 \end{pmatrix} \right) = 2 \pi \cdot {}_1F_1 \left(\frac{1}{2}, 1, z_1 \right)\ , \]
so it is sufficient to compute the hypergeometric function of a scalar argument, which is described in \cite{abramowitz1964}.

\subsection{Multiplication}
For two given Bingham densities, we want to obtain their product. This product is used for Bayesian inference involving Bingham distributions. The result presented below yields a convenient way to calculate the product of Bingham distributions.
\begin{lemma}
Bingham distributions are closed under multiplication with renormalization.
\end{lemma}
\begin{proof}
Consider two Bingham distributions
\[ f_1(\vec{x}) = F_1 \cdot \exp ( \vec{x}^T \mat{M}_1\, \mat{Z}_1\, \mat{M}_1^T \vec{x} )  \]
and 
\[ f_2(\vec{x}) = F_2 \cdot \exp ( \vec{x}^T \mat{M}_2\, \mat{Z}_2\, \mat{M}_2^T \vec{x} ) \ . \] 
Then 
\begin{align*} 
f_1(\vec{x}) \cdot f_2(\vec{x}) 
&= F_1 F_2 \cdot \exp (\vec{x}^T ( \underbrace{\mat{M}_1 \mat{Z}_1 \mat{M}_1^T + \mat{M}_2 \mat{Z}_2 \mat{M}_2^T)}_{ =: \mat{C}} \vec{x}) \\
&\propto F \cdot \exp ( \vec{x}^T \mat{M}\, \mat{Z}\, \mat{M}^T \vec{x} )
\end{align*}
with $F$ as the new normalization constant after renormalization, $\mat{M}$ are the unit eigenvectors of $\mat{C}$, $\mat{D}$ has the eigenvalues of $\mat{C}$ on the diagonal (sorted in ascending order) and $\mat{Z} = \mat{D} - \mat{D}_{dd}\mat{I}_{d \times d}$ where $ \mat{D}_{dd}$ refers to the bottom right entry of $\mat{D}$, i.\,e., the largest eigenvalue.
\end{proof}

\subsection{Estimation of Parameters}
Estimating parameters for the Bingham distribution is not only motivated by the need to estimate noise parameters from samples. It also plays a crucial role in the prediction process when computing the composition of two Bingham random vectors. This procedure is based on matching covariance matrix. Be aware that although the Bingham distribution is only defined on $S_{d-1}$, we can still compute its covariance in $\mathbb{R}^d$. Thus, we will present both the computation of the covariance matrix of a Bingham distributed random vector and the computation of parameters for a Bingham distribution with a given covariance (which could originate from an arbitrary distribution on the hypersphere).


The maximum likelihood estimate for the parameters $(\mat{M},\mat{Z})$ of a Bingham distribution can be obtained as described in \cite{bingham1974}. $\mat{M}$ can be obtained as the matrix of eigenvectors of the covariance $\mat{S}$ with eigenvalues $\omega_1 \leq \omega_2$. In other words, $\mat{M}$ can be found as the eigendecomposition of $\mat{S} = \mat{M} \cdot \diag (\omega_1, \omega_2) \cdot \mat{M}^T$. To calculate $\mat{Z}$, the equations
\[ \frac{\frac{\partial}{\partial z_i} {}_1F_1 \left( \frac{1}{2}, 1, \begin{pmatrix} z_1 & 0 \\ 0 & z_2 \end{pmatrix} \right)}{{}_1F_1 (\frac{1}{2}, 1, z_1 )}  = \omega_i, \quad i=1,2 \]
have to be solved under the constraint $z_2=0$, which is justified by the argumentation above and used to simplify the computation.  This operation is performed numerically.

Conversely, for a given a Bingham distribution $(\mat{M},\mat{Z})$, the covariance matrix can be calculated according to 
\begin{align*} 
\mat{S} &=  \mat{M} \cdot \diag (\omega_1, \omega_2) \cdot \mat{M}^T \\
&=  \mat{M} \cdot \diag \left( \frac{1}{F} \frac{\partial F}{\partial z_1}, \frac{1}{F} \frac{\partial F}{\partial z_2} \right) \cdot \mat{M}^T  \\
&=\sum_{i=1}^2 \frac{1}{F} \frac{\partial F}{\partial z_i} \mat{M}(:,i) \mat{M}(:,i)^T \ ,
\end{align*}
where $\mat{M}(:,i)$ refers to the $i$-th column of $\mat{M}$ \cite{Glover13}. Thus, any Bingham distribution is uniquely defined by its covariance matrix and vice versa. The following Lemma simplifies the computation of partial derivatives of a confluent hypergeometric function of a $2\times 2$ matrix argument, which is used in computation of the covariance matrix as derived above.

\begin{lemma}
\label{lemma:partialderivatives}
For $d=2$, the partial derivatives
\[  \frac{\partial}{\partial z_i} {}_1F_1 \left( \frac{1}{2}, 1, \begin{pmatrix} z_1 & 0 \\ 0 & z_2 \end{pmatrix}\right) \ , \quad i=1,2 \]
can be reduced to hypergeometric functions of scalar argument.
\end{lemma}
\begin{proof}
See \ref{sec:proofpartialderivatives}.
\end{proof}

\subsection{Composition}
Now, we want to derive the composition of Bingham distributed random variables, which is the directional analogue to adding random variables. This operation can, for example, be used to disturb an uncertain Bingham-distributed system state with Bingham-distributed noise, similar to using a convolution to disturb a probability distribution on $\mathbb{R}$ with additive noise. First, we define a composition of individual points on the hypersphere $S_{d-1}$, which we then use to derive the composition of Bingham distributions.


The composition of two Bingham distributions depends on the interpretation of the unit vectors, for example as complex numbers or quaternions. We assume that a composition function
\[ \oplus: S_{d-1} \times S_{d-1} \to S_{d-1} \]
is given. The function $\oplus$ has to be compatible with $180^\circ$ degree symmetry, i.e., 
\[ \pm(x \oplus y) = \pm( (-x) \oplus y ) = \pm( x \oplus (-y) ) = \pm( (-x) \oplus (-y) ) \]
for all $x,y \in S_{d-1}$. Furthermore, we require the quotient $(S_{d-1} / \{ \pm 1 \}, \oplus)$ to have an algebraic group structure. This guarantees associativity, the existence of an identity element, and the existence of inverse elements.

In the complex case, we interpret $S_1 \subset \mathbb{R}^2$ as unit vectors in $\mathbb{C}$, where the first dimension is the real part and the second dimension the imaginary part. In this interpretation, the Bingham distributions can be understood as a distribution on a subset of the complex plane, namely the unit circle.
\begin{definition}
The composition function $\oplus$ is defined to be complex multiplication, i.e.,
\[ \begin{pmatrix} x_1 \\ x_2\end{pmatrix} \oplus \begin{pmatrix} y_ 1\\ y_2\end{pmatrix} = \begin{pmatrix} x_1y_1 - x_2 y_2\\ x_1 y_2 + x_2 y_1 \end{pmatrix} \]
analogous to
\[ (x_1 + i x_2) \cdot (y_1 + i y_2) = (x_1 y_1 - x_2 y_2) + i (x_1 y_2 + x_2 y_1) \ . \]
\end{definition}
Since we only consider unit vectors, the composition $\oplus$ is equivalent to adding the angles of both complex numbers when they are represented in polar form. The identity element is $\pm 1$ and the inverse element for $(x_1,x_2)$ is the complex conjugate $\pm(x_1,-x_2)$.

Unfortunately, the Bingham distribution is not closed under this kind of composition. That is, the resulting random vector is not Bingham distributed. Thus, we propose a technique to approximate a Bingham distribution to the composed random vector. The composition of two Bingham distributions $f_\mat{A}$ and $f_\mat{B}$ is calculated by considering the composition of their covariance matrices $\mat{A}, \mat{B}$ and estimating the parameters of $f_\mat{C}$ based on the resulting covariance matrix. Composition of covariance matrices can be derived from the composition of random vectors.

\begin{lemma}
Let $f_\mat{A}$ and $f_\mat{B}$ be Bingham distributions with covariance matrices \label{lemma:composition} 
\[\mat{A} = \begin{pmatrix}a_{11} & a_{12} \\ * & a_{22}\end{pmatrix} \text{ and } \mat{B}=\begin{pmatrix}b_{11} & b_{12} \\ * & b_{22}\end{pmatrix}\ , \]
respectively. Let $\vec{x},\vec{y} \in S_1 \subset \mathbb{R}^2$ be independent random vectors distributed according to $f_\mat{A}$ and $f_\mat{B}$. Then the covariance
 \[ \mat{C} =  \begin{pmatrix}c_{11} & c_{12} \\ * & c_{22}\end{pmatrix} :=  \Cov (\vec{x} \oplus \vec{y}) \]
of the composition is given by
\begin{align*}
c_{11} =& a_{11} b_{11} - 2 a_{12} b_{12} + a_{22} b_{22}\ , \\
c_{12} =& a_{11} b_{12} - a_{12} b_{22} + a_{12} b_{11} - a_{22} b_{12}\ , \\
c_{22} =& a_{11} b_{22} + 2 a_{12} b_{12} + a_{22} b_{11}\ .\\
\end{align*}
\end{lemma}
\begin{proof}
See \ref{sec:proofcomposition}.
\end{proof}

Based on $\mat{C}$, maximum likelihood estimation is used to obtain the parameters $\mat{M}$ and $\mat{Z}$ of the uniquely defined Bingham distribution with covariance $\mat{C}$ as described above. This computation can be done in an efficient way, because the solution of the equation involving the hypergeometric function is the only part which is not given in closed form. This does not present a limitation to the proposed algorithm, because there are many efficient ways for the computation of the confluent hypergeometric function of a scalar argument \cite{luke1977, muller2001}.

%% file: filter.tex
The techniques presented in the preceding section can be applied to derive a filter based on the Bingham distribution. The system model is given by
\[ \vec{x}_{k+1} = \vec{x}_k \oplus \vec{w}_k\ ,\]
where $\vec{w}_k$ is Bingham distributed noise. The measurement model is given by
\[ \vec{z}_k = \vec{x}_k \oplus \vec{v}_k\ ,\]
where $\vec{v}_k$ is Bingham distributed noise and $\vec{x}_k$ is an uncertain Bingham distributed system state. Intuitively, this means that both system and measurement model are the identity disturbed by Bingham distributed noise. Note that $\vec{w}_k$ and $\vec{v}_k$ can include a constant offset. For example $\vec{w}_k$ could include a known angular velocity. Alternatively, to avoid dealing with biased noise distributions, a rotation may be applied to $\vec{x}_k$ first and unbiased noise added subsequently.

The predicted and estimated distributions at time $k$ are described by their parameter matrices $(\mat{M}_k^p, \mat{Z}_k^p)$ and $(\mat{M}_k^e, \mat{Z}_k^e)$ respectively. The noise distributions at time $k$ are described by $(\mat{M}_k^w, \mat{Z}_k^w)$ and $(\mat{M}_k^v,\mat{Z}_k^v)$.

\subsection{Prediction}
The prediction can be calculated according to
\[ (\mat{M}_{k+1}^p, \mat{Z}_{k+1}^p) = \text{composition}((\mat{M}_k^e, \mat{Z}_k^e),(\mat{M}_k^w, \mat{Z}_k^w)) \ , \]
which uses the previously introduced composition operation to disturb the estimate with the system noise.

\begin{algorithm}[tb]
\KwIn{estimate $\mat{M}_k^e, \mat{Z}_k^e$, noise $\mat{M}_k^w, \mat{Z}_k^w$}
\vspace{4mm}
\KwOut{prediction $\mat{M}_{k+1}^p, \mat{Z}_{k+1}^p$}
\vspace{4mm}
\tcc{calculate covariance matrices $\mat{A},\mat{B}$}
$\mat{A} = \sum_{i=1}^d \frac{1}{F} \frac{\partial F}{\partial z_i} \mat{M}_k^e(:,i) \mat{M}_k^e(:,i)^T$\; 
$\mat{B} = \sum_{i=1}^d \frac{1}{F} \frac{\partial F}{\partial z_i} \mat{M}_k^w(:,i) \mat{M}_k^w(:,i)^T$\;
\tcc{calculate $\mat{C}$ according to Lemma \ref{lemma:composition}}
$c_{11} = a_{11} b_{11} - 2 a_{12} b_{12} + a_{22} b_{22}$\;
$c_{12} = a_{11} b_{12} - a_{22} b_{12} - a_{12} b_{22} + a_{12} b_{11}$\; 
$c_{22} = a_{11} b_{22} + 2 a_{12} b_{12} + a_{22} b_{11}$\;
$\mat{C} = \begin{pmatrix} c_{11} & c_{12} \\ c_{12} & c_{22} \end{pmatrix}$\;
\tcc{calculate $\mat{M}_{k+1}^p, \mat{Z}_{k+1}^p$ based on $\mat{C}$}
$\mat{M}_{k+1}^p, \mat{Z}_{k+1}^p \gets \text{MLE}(\mat{C})$\; 
\label{algo:predict}
\caption{Prediction}
\end{algorithm}

\subsection{Update}
Given a measurement $\hat{\vec{z}}$, we can calculate the updated distribution according to Bayes' rule
\[ f(\mat{M}_k, \mat{Z}_k | \hat{\vec{z}}) = c \cdot f(\hat{\vec{z}} | \mat{M}_k, \mat{Z}_k) \cdot f(\mat{M}_k, \mat{Z}_k) \]
with some normalization constant $c$, which yields the update procedure
\[ (\mat{M}_k^e, \mat{Z}_k^e) = \text{multiply}((\mat{M}, \mat{Z}_k^e), (\mat{M}_k^p, \mat{Z}_k^p)) \]
with $\mat{M} = (\bar{\hat{\vec{z}}} \oplus \mat{M}_k^v) $, where $\bar{\vec{a}}$ indicates the complex conjugate of $\vec{a}$ and $\oplus$ is evaluated for each column of $\mat{M}_k^v$.

\begin{algorithm}[tb]
\KwIn{prediction $\mat{M}_k^p, \mat{Z}_k^p$, noise $\mat{M}_k^v, \mat{Z}_k^v$, measurement $\hat{\vec{z}}_k$}
\vspace{4mm}
\KwOut{estimate $\mat{M}_k^e, \mat{Z}_k^e$}
\vspace{4mm}
\tcc{rotate noise according to measurement}
$\mat{M} \gets \begin{pmatrix} \bar{\hat{\vec{z}}} \oplus \mat{M}_k^v \end{pmatrix}$\; 
\tcc{multiply with prior distribution}
$(\mat{M}_k^e, \mat{Z}_k^e) \gets \text{multiply}((\mat{M}, \mat{Z}_k^v)), (\mat{M}_k^p, \mat{Z}_k^p))$\;
\label{algo:update}
\caption{Update}
\end{algorithm}

%% file: evaluation.tex
The proposed filter was evaluated in simulations. In this section, all angles are given in radians unless specified differently.

For comparison, we implemented a one-dimensional Kalman filter \cite{kalman1960}. A traditional one-dimensional Kalman filter has two issues when confronted with our situation. First, it does not take the circular nature of the problem into account. Second, it does not handle $180^\circ$ symmetry. We can circumvent both issues by restricting the estimate $x_k$ according to $0 \leq x_k \leq \pi$ and by shifting the measurement, so that $|x_k - \hat{z}_k| \leq \frac{\pi}{2}$ is satisfied.

%
%

In our example, we consider the estimation of an axis in robotics. This could be the axis of a symmetric rotor blade  or any other robotic joint with $180^\circ$ symmetry. We use the initial estimate with mode $(0,1)^T$
\[ \mat{M}_0^e = \begin{pmatrix} 1 & 0 \\ 0 & 1 \end{pmatrix}, \quad \mat{Z}_0^e = \begin{pmatrix} -1 & 0 \\ 0  & 0 \end{pmatrix} \ , \]
the system noise with mode $(1,0)^T$
\[ \mat{M}_k^w = \begin{pmatrix} 0 & 1 \\ 1 & 0 \end{pmatrix}, \quad \mat{Z}_k^w = \begin{pmatrix} -200 & 0 \\ 0  & 0 \end{pmatrix} \ , \]
and the measurement noise with mode $(1,0)^T$
\[ \mat{M}_k^v = \begin{pmatrix} 0 & 1 \\ 1 & 0 \end{pmatrix}, \quad \mat{Z}_k^v = \begin{pmatrix} -3 & 0 \\ 0  & 0 \end{pmatrix} \ . \]
The true initial state is given by $(1, 0)^T$, i.\,e., the initial estimate with mode $(0,1)^T$ is very poor. 
The initial estimate for the Kalman filter is given by 
\[ x_0^e = \atan2(\text{mode} (\mat{M}_0^e)) = \atan2(1,0) = \frac{\pi}{2} \]
and the noise means are
\[ \mu_k^w = \mu_k^v = \atan2(0,1) = 0 \ . \]

The covariance matrices for the Kalman filter are obtained by sampling the Bingham noise parameters and calculating the empirical covariance from the samples. This yields
\[ C_0^e = 0.5956, \quad C_k^w = 0.0027, \quad C_k^v=0.2836 \ , \]
which is equivalent to standard deviations of $44^\circ$ for the first time step, $3^\circ$ for the system noise and $30^\circ$ for the measurement noise.

\begin{figure*}[t]
\centering
\begin{subfigure}{0.7\textwidth}
\centering
\includegraphics[width=0.6\textwidth]{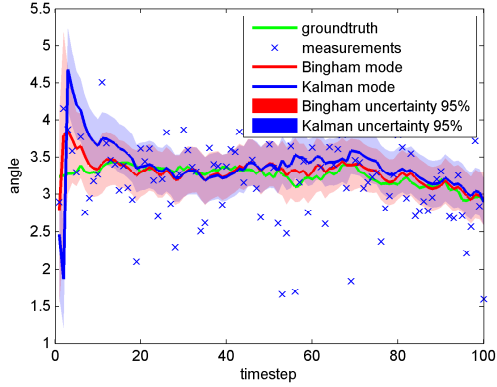}
\caption{Ground truth and estimate.}
\end{subfigure}
\begin{subfigure}{0.7\textwidth}
\centering
\includegraphics[width=0.6\textwidth]{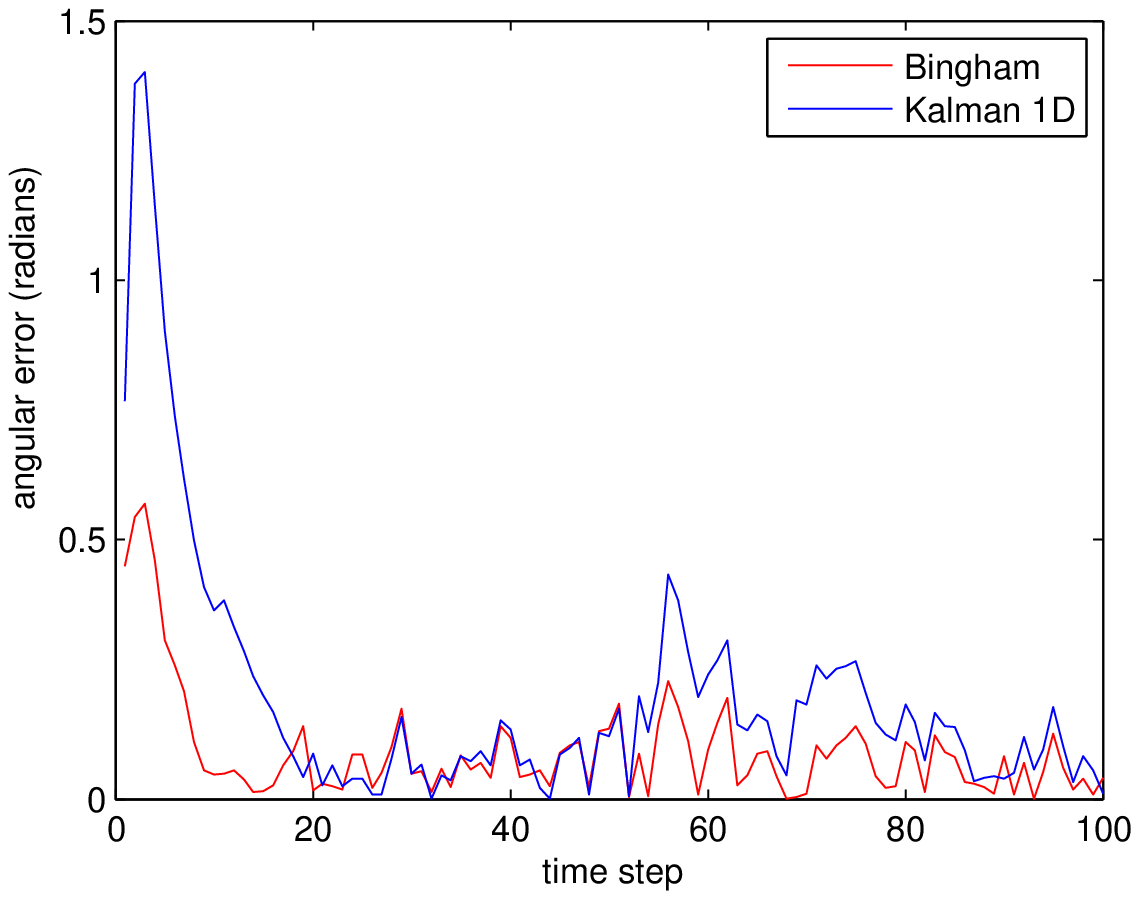}
\caption{Angular error.}
\end{subfigure}
\caption{An example run of Bingham and Kalman filter.}
\label{fig:examplerun}
\end{figure*}

We simulate the system for a duration of $k_{\max} = 100$ time steps. An example run is depicted in Fig. \ref{fig:examplerun}. In addition to the mode of the estimate, we plot the $95 \%$ confidence interval, which is equivalent to the $2 \sigma$ bounds in the case of the Kalman filter.


For evaluation, we consider the angular RMSE which is given by
\[ \sqrt{ \frac{1}{k_{\max}} \sum_{k=1}^{k_{\max}} (e_k)^2 } \]
with angular error 
\[ e_k = \min(\measuredangle (\vec{x}_k^\text{true}, \text{mode}(\mat{M}_k^e)) , \pi - \measuredangle (\vec{x}_k^\text{true}, \text{mode}(\mat{M}_k^e)) \]
at time step $k$. Obviously, $0 \leq e_k \leq \frac{\pi}{2}$ holds, which is consistent with our assumption of $180^\circ$ symmetry.

The presented results are based on 1000 Monte Carlo runs. Even though our filter is computationally more demanding than a Kalman filter, it is still fast enough for real-time applications. On a standard laptop, our non-optimized implementation in MATLAB needs approximately $\SI{60}{\milli\second}$ for one time step (prediction and update), which could be significantly improved by a faster evaluation of the hypergeometric function. 
In Fig. \ref{fig:montecarlo-scatter}, we plot the error of our filter against the error of the Kalman filter for all runs. The proposed filter outperforms the Kalman filter in most cases, which is also true for the mean angular error in every time step as shown in Fig. \ref{fig:montecarlo-meanerror}. In particular, the significantly faster rate of convergence of the proposed filter is evident. This superiority is due to the reasons listed in the introduction. The use of the Gaussian distribution, which does not consider the problem geometry leads to suboptimal results compared to the proposed approach based on the Bingham distribution. 

In Fig. \ref{fig:montecarlo-meanerror}, we also show a comparison with a filter based on the wrapped normal distribution (denoted WN), which we previously published in \cite{kurz2013} and modified for the $180^\circ$ case. The angular error of both filters is almost indistinguishable. However, unlike the proposed filter based on the Bingham distribution, the previously published filter cannot easily be generalized to higher dimensions.


\begin{figure*}[t]
\centering
\begin{subfigure}{0.45\textwidth}
\centering
\includegraphics[height=5cm]{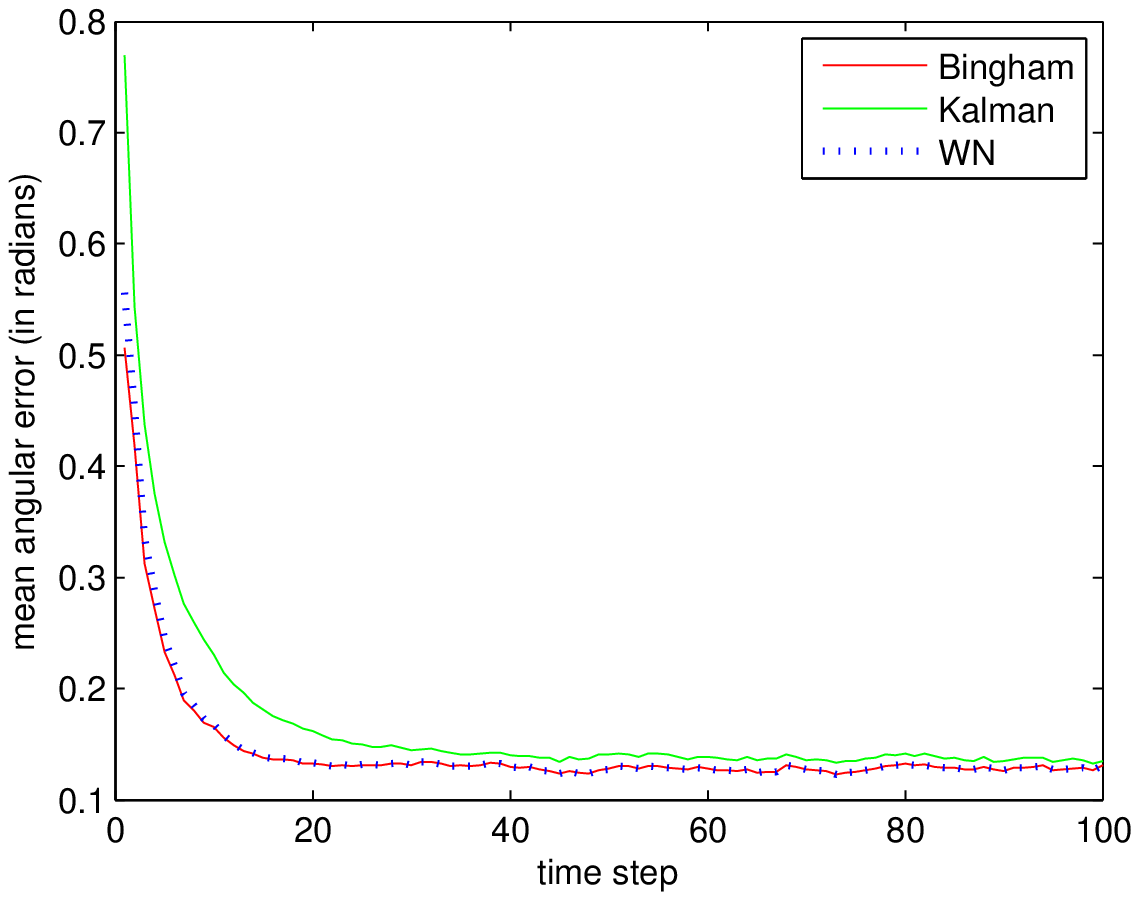}
\caption{Time steps $1 \leq k \leq 100$.}
\end{subfigure}
\begin{subfigure}{0.45\textwidth}
\centering
\includegraphics[height=5cm]{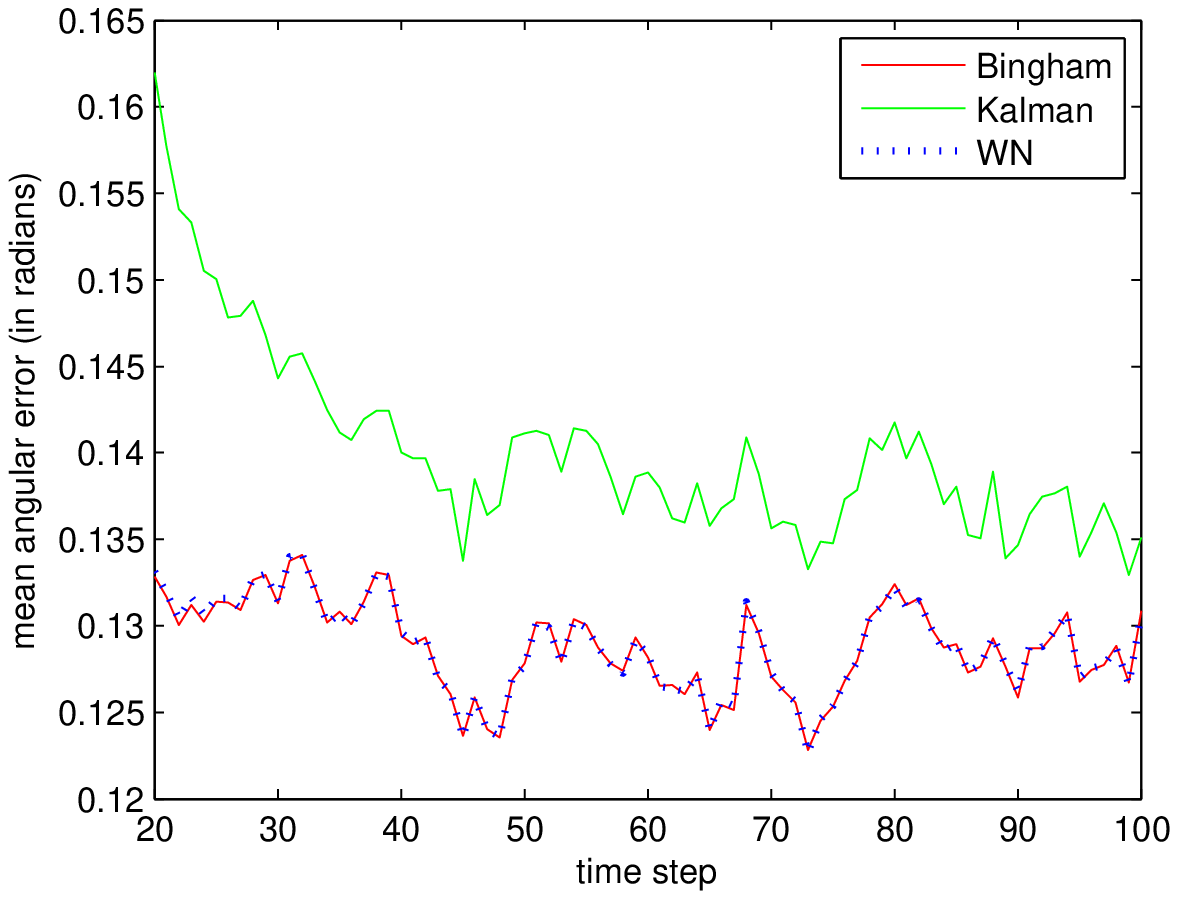}
\caption{Time steps $20 \leq k \leq 100$.}
\end{subfigure}
\caption{Results of 1000 Monte Carlo runs. We show the mean error at every time step across all runs for the proposed Bingham filter, a Kalman \cite{kalman1960} filter and a filter based on the wrapped normal (WN) distribution \cite{kurz2013}. Because the initial error is large as a result of the poor initial estimate, we show two plots of different time intervals.}
\label{fig:montecarlo-meanerror}
\end{figure*}

\begin{figure}[t]
\centering
\includegraphics[width=0.6\textwidth]{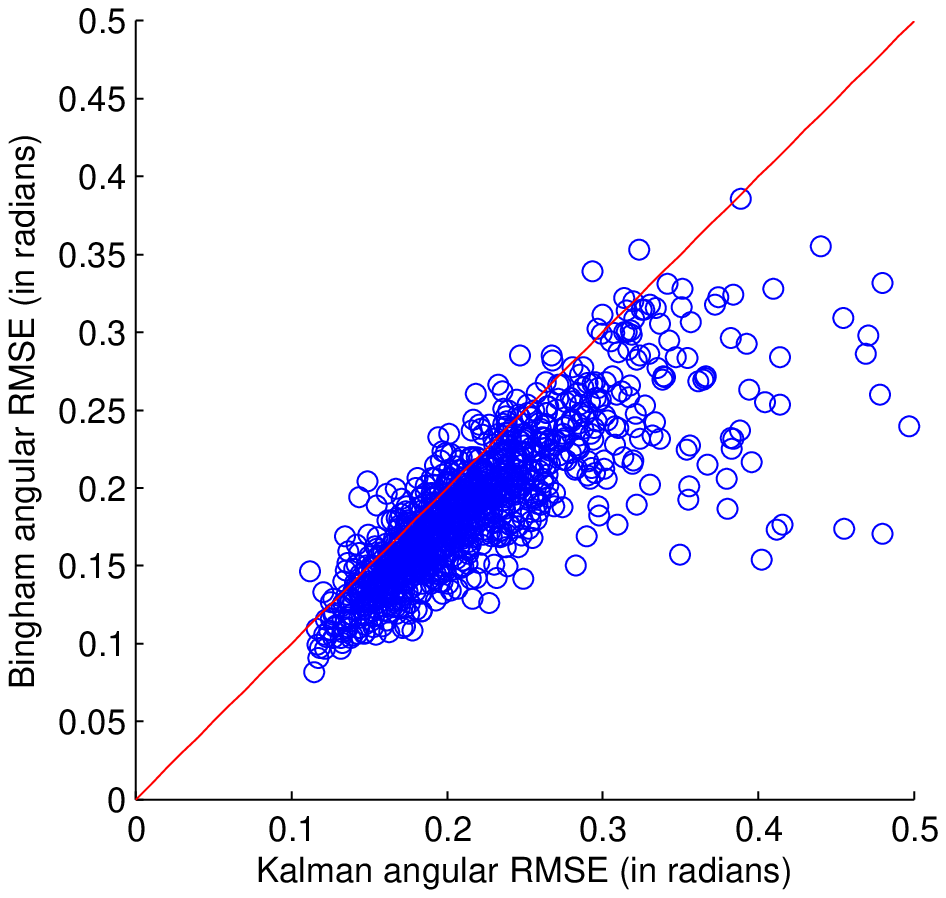}
\caption{Results of 1000 Monte Carlo runs. Each sample represents one run. Samples below the red line indicate that the proposed filter has performed better, samples above the red line indicate that the Kalman filter has performed better.}
\label{fig:montecarlo-scatter}
\end{figure}


%% file: conclusion.tex
We have presented a recursive filter based on the Bingham distribution. It can be applied to circular estimation problems with $180^\circ$ symmetry. Our simulations have shown the superiority of the presented approach compared to the traditional solution of modifying a Kalman filter for the circular setting.

Future work will focus on recursive 3D pose estimation using Bingham distribution. This can be achieved by applying the presented methods in the four-dimensional case for estimating quaternions. Open challenges include an efficient estimator of the Bingham parameters based on available data. This makes an efficient evaluation of the confluent hypergeometric function necessary.

%% file: appendix.tex
\section{Proof of Lemma \ref{lemma:partialderivatives}.}
\label{sec:proofpartialderivatives}
\noindent
We use the identities
\[ {}_1F_1 \left(\frac{1}{2}, 1, \begin{pmatrix} z_1 & 0 \\ 0 & z_2 \end{pmatrix} \right) =  \exp(z_2) \cdot {}_1F_1 \left( \frac{1}{2},1,z_1-z_2 \right) \]
and
\[ \frac{\partial}{\partial z} {}_1F_1 (a,b,z) = \frac{a}{b} {}_1F_1 (a+1, b+1, z) \ . \]
This yields
\begin{align*}
\frac{\partial}{\partial z_1} \frac{F}{|S_{d-1}|}
& = \frac{\partial}{\partial z_1} {}_1F_1 \left( \frac{1}{2}, 1, \begin{pmatrix} z_1 & 0 \\ 0 & z_2 \end{pmatrix}\right)\\ 
& = \exp(z_2) \frac{\partial}{\partial z_1} {}_1F_1 \left( \frac{1}{2}, 1, z1-z2 \right) \\ 
& = \exp(z_2) \frac{1}{2} {}_1F_1 \left( \frac{3}{2}, 2, z1-z2 \right)
\end{align*}
and
\begin{align*}
\frac{\partial}{\partial z_2} \frac{F}{|S_{d-1}|}
= & \frac{\partial}{\partial z_2} {}_1F_1 \left( \frac{1}{2}, 1, \begin{pmatrix} z_1 & 0 \\ 0 & z_2 \end{pmatrix}\right) \\
= &\frac{\partial}{\partial z_2} \left( \exp(z_2) {}_1F_1 \left( \frac{1}{2}, 1, z_1 - z_2 \right) \right) \\
= &\exp(z_2) {}_1F_1 \left( \frac{1}{2}, 1, z_1 - z_2 \right) + \exp(z_2) \frac{\partial}{\partial z_2} {}_1F_1 \left( \frac{1}{2}, 1, z_1 - z_2 \right) \\
= &\exp(z_2) {}_1F_1 \left( \frac{1}{2}, 1, z_1 - z_2 \right) - \exp(z_2) \frac{1}{2} {}_1F_1 \left( \frac{3}{2}, 2, z_1 - z_2 \right) \\
= &\exp(z_2) \Bigg( {}_1F_1 \left( \frac{1}{2}, 1, z_1 - z_2 \right) - \frac{1}{2} {}_1F_1 \left( \frac{3}{2}, 2, z_1 - z_2 \right) \Bigg)
\end{align*}
\hfill\qedsymbol

\section{Proof of Lemma \ref{lemma:composition}.}
\label{sec:proofcomposition}
\noindent

The covariance of the composition
\begin{align*}
\mat{C} =& \Cov (\vec{x} \oplus \vec{y}) \\
=& \Cov \left( \begin{pmatrix} x_1y_1 - x_2 y_2\\ x_1 y_2 + x_2 y_1 \end{pmatrix} \right) \\
=& \begin{pmatrix}
\Var (x_1y_1 - x_2 y_2) & \Cov(x_1y_1 - x_2 y_2,  x_1 y_2 + x_2 y_1)\\
* & \Var( x_1 y_2 + x_2 y_1)
\end{pmatrix}
\end{align*}
can be obtained by calculating the matrix entries individually. For the first entry we get
\begin{align}
c_{11} =& \Var (x_1y_1 - x_2 y_2) \notag \\
=& \E ((x_1y_1 - x_2 y_2)^2) - (\E(x_1y_1 - x_2 y_2))^2 \notag \\
=& E(x_1^2y_1^2 - 2 x_1y_1 x_2y_2 + x_2^2y_2^2) 
- (\E(x_1y_1) - \E(x_2y_2))^2 \label{eq:independence} \\
=& \E(x_1^2)\E(y_1^2) - 2 \E(x_1x_2)\E(y_1y_2) + \E(x_2^2)\E(y_2^2) \label{eq:linearity} \\
&- (\underbrace{\E(x_1)}_{0} \underbrace{\E(y_1)}_{0} - \underbrace{\E(x_2)}_{0} \underbrace{\E(y_2)}_{0})^2 \label{eq:zero} \\
=& a_{11} b_{11} - 2 a_{12} b_{12} + a_{22} b_{22} \notag .
\end{align}
We use independence of $\vec{x}$ and $\vec{y}$ in (\ref{eq:independence}), linearity of the expectation value in (\ref{eq:linearity}) and symmetry of the Bingham in (\ref{eq:zero}).
Analogously we calculate
\begin{align*}
c_{22} =& a_{11} b_{22} + 2 a_{12} b_{12} + a_{22} b_{11} \ .
\end{align*}
The off-diagonal entry can be calculated similarly
\begin{align*}
c_{12} =& \Cov(x_1y_1 - x_2 y_2,  x_1 y_2 + x_2 y_1) \\
=& \E((x_1y_1 - x_2 y_2) \cdot (x_1 y_2 + x_2 y_1)) 
- \E(x_1y_1 - x_2 y_2)\cdot \E(x_1 y_2 + x_2 y_1)\\
=& \E (x_1^2y_1y_2 - x_1x_2y_2^2 + x_1x_2y_1^2 - x_2^2y_1y_2 )
- (\E(x_1)\E(y_1) - \E(x_2) \E(y_2)) \\ 
&\cdot (\E(x_1) \E( y_2) + \E(x_2) \E(y_1)) \\
=& a_{11} b_{12} - a_{12} b_{22} + a_{12} b_{11} - a_{22} b_{12} \ .
\end{align*}
\hfill\qedsymbol

%% file: main.bbl
\begin{thebibliography}{10}
\providecommand{\url}[1]{#1}
\csname url@samestyle\endcsname
\providecommand{\newblock}{\relax}
\providecommand{\bibinfo}[2]{#2}
\providecommand{\BIBentrySTDinterwordspacing}{\spaceskip=0pt\relax}
\providecommand{\BIBentryALTinterwordstretchfactor}{4}
\providecommand{\BIBentryALTinterwordspacing}{\spaceskip=\fontdimen2\font plus
\BIBentryALTinterwordstretchfactor\fontdimen3\font minus
  \fontdimen4\font\relax}
\providecommand{\BIBforeignlanguage}[2]{{%
\expandafter\ifx\csname l@#1\endcsname\relax
\typeout{** WARNING: IEEEtran.bst: No hyphenation pattern has been}%
\typeout{** loaded for the language `#1'. Using the pattern for}%
\typeout{** the default language instead.}%
\else
\language=\csname l@#1\endcsname
\fi
#2}}
\providecommand{\BIBdecl}{\relax}
\BIBdecl

\bibitem{kalman1960}
R.~E. Kalman, ``{A} {N}ew {A}pproach to {L}inear {F}iltering and {P}rediction
  {P}roblems,'' \emph{Transactions of the {ASME} Journal of Basic Engineering},
  vol.~82, pp. 35--45, 1960.

\bibitem{julier2004}
S.~Julier and J.~Uhlmann, ``Unscented filtering and nonlinear estimation,''
  \emph{Proceedings of the IEEE}, vol.~92, no.~3, pp. 401--422, mar 2004.

\bibitem{julier2007}
S.~Julier and J.~LaViola, ``On kalman filtering with nonlinear equality
  constraints,'' \emph{IEEE Transactions on Signal Processing}, vol.~55, no.~6,
  pp. 2774--2784, 2007.

\bibitem{Hertzberg2013}
C.~Hertzberg, R.~Wagner, U.~Frese, and L.~Schröder, ``Integrating generic
  sensor fusion algorithms with sound state representations through
  encapsulation of manifolds,'' \emph{Information Fusion}, vol.~14, no.~1, pp.
  57 -- 77, 2013.

\bibitem{shiryaev1995}
A.~N. Shiryaev, \emph{Probability}, 2nd~ed.\hskip 1em plus 0.5em minus
  0.4em\relax Springer, 1995.

\bibitem{bingham1974}
C.~Bingham, ``An antipodally symmetric distribution on the sphere,'' \emph{The
  Annals of Statistics}, vol.~2, no.~6, pp. 1201--1225, Nov. 1974.

\bibitem{mardia1999}
K.~V. Mardia and P.~E. Jupp, \emph{Directional Statistics}, 1st~ed.\hskip 1em
  plus 0.5em minus 0.4em\relax Wiley, 1999.

\bibitem{jammalamadaka2001}
S.~R. Jammalamadaka and A.~Sengupta, \emph{Topics in Circular
  Statistics}.\hskip 1em plus 0.5em minus 0.4em\relax World Scientific Pub Co
  Inc, 2001.

\bibitem{krinidis2006}
S.~Krinidis and V.~Chatzis, ``Frequency-based object orientation and scaling
  determination,'' in \emph{2006 {IEEE} International Symposium on Circuits and
  Systems, 2006. {ISCAS} 2006. Proceedings}, May 2006, p. 4 pp.

\bibitem{azmani2009}
M.~Azmani, S.~Reboul, J.-B. Choquel, and M.~Benjelloun, ``A recursive fusion
  filter for angular data,'' in \emph{2009 {IEEE} International Conference on
  Robotics and Biomimetics ({ROBIO)}}, Dec. 2009, pp. 882 --887.

\bibitem{kurz2013}
G.~Kurz, I.~Gilitschenski, and U.~D. Hanebeck, ``Recursive nonlinear filtering
  for angular data based on circular distributions,'' in \emph{Proceedings of
  the 2013 American Control Conference ({ACC} 2013) (to appear)}, Washington
  {D.C.}, {USA}, Jun. 2013.

\bibitem{kent1987}
J.~T. Kent, ``Asymptotic expansions for the bingham distribution,''
  \emph{Journal of the Royal Statistical Society. Series C (Applied
  Statistics)}, vol. 36 (2), pp. 139--144, 1987.

\bibitem{jupp1979}
P.~E. Jupp and K.~V. Mardia, ``Maximum likelihood estimators for the matrix von
  mises-fisher and bingham distributions,'' \emph{Annals of Statistics}, vol. 7
  (3), pp. 599--606, 1979.

\bibitem{kunze2004}
K.~Kunze and H.~Schaeben, ``\BIBforeignlanguage{English}{The bingham
  distribution of quaternions and its spherical radon transform in texture
  analysis},'' \emph{\BIBforeignlanguage{English}{Mathematical Geology}},
  vol.~36, pp. 917--943, 2004.

\bibitem{glover2011}
J.~Glover, R.~Rusu, and G.~Bradski, ``Monte carlo pose estimation with
  quaternion kernels and the bingham distribution,'' in \emph{Proceedings of
  Robotics: Science and Systems}, Los Angeles, {CA}, {USA}, Jun. 2011.

\bibitem{Glover13}
\BIBentryALTinterwordspacing
J.~Glover, ``libbingham bingham statistics library,'' 2013. [Online].
  Available: \url{http://code.google.com/p/bingham/}
\BIBentrySTDinterwordspacing

\bibitem{kume2005}
A.~Kume and A.~T.~A. Wood, ``Saddlepoint approximations for the bingham and
  {FisherâBingham} normalising constants,'' \emph{Biometrika}, vol.~92,
  no.~2, pp. 465--476, 2005.

\bibitem{herz1955}
C.~S. Herz, ``Bessel functions of matrix argument,'' \emph{Annals of
  Mathematics}, vol.~61, no.~3, pp. 474--523, 1955.

\bibitem{koev2006}
P.~Koev and A.~Edelman, ``The efficient evaluation of the hypergeometric
  function of a matrix argument,'' \emph{Math. Comp.}, vol.~75, pp. 833--846,
  2006.

\bibitem{muirhead1982}
R.~Muirhead, \emph{Aspects of multivariate statistical theory}, ser. Wiley
  series in probability and mathematical statistics: Probability and
  mathematical statistics.\hskip 1em plus 0.5em minus 0.4em\relax Wiley, 1982.

\bibitem{abramowitz1964}
M.~{Abramowitz} and I.~A. {Stegun}, \emph{Handbook of Mathematical Functions
  with Formulas, Graphs, and Mathematical Tables}, ninth dover printing, tenth
  gpo printing~ed.\hskip 1em plus 0.5em minus 0.4em\relax New York: Dover,
  1964.

\bibitem{luke1977}
Y.~Luke, \emph{Algorithms for the computation of mathematical functions}, ser.
  Computer science and applied mathematics.\hskip 1em plus 0.5em minus
  0.4em\relax Academic Press, 1977.

\bibitem{muller2001}
K.~E. Muller, ``\BIBforeignlanguage{en}{Computing the confluent hypergeometric
  function, m(a,b,x)},'' \emph{\BIBforeignlanguage{en}{Numerische Mathematik}},
  vol.~90, no.~1, pp. 179--196, 2001.

\end{thebibliography}
